\documentclass{IEEEtran}

\usepackage{amsmath,amssymb} \interdisplaylinepenalty=2500
\usepackage{cite}
\usepackage{graphicx}
\usepackage{color}

\newtheorem{definition}{Definition}
\newtheorem{theorem}{Theorem}
\newtheorem{proposition}[theorem]{Proposition}

\newtheorem{lemma}[theorem]{Lemma}

\newtheorem{myexample}{Example}
\newenvironment{example}{\myexample}{\qed\endmyexample}

\def\qed{\endIEEEproof}

\newcommand{\Fq}{\mathbb{F}_q}

\newcommand{\calA}{\mathcal{A}}
\newcommand{\calB}{\mathcal{B}}

\newcommand{\calE}{\mathcal{E}}

\newcommand{\calG}{\mathcal{G}}

\newcommand{\calN}{\mathcal{N}}

\newcommand{\calS}{\mathcal{S}}
\newcommand{\calT}{\mathcal{T}}
\newcommand{\calU}{\mathcal{U}}
\newcommand{\calV}{\mathcal{V}}

\DeclareMathOperator{\indegree}{\sf indegree}
\DeclareMathOperator{\outdegree}{\sf outdegree}
\DeclareMathOperator{\tail}{\sf tail}
\DeclareMathOperator{\head}{\sf head}
\DeclareMathOperator{\K}{\sf mincut}

\title{Robust Network Coding in the Presence of Untrusted Nodes}

\author{Da Wang, Danilo Silva and Frank R. Kschischang%
\thanks{The work of D. Wang was supported by NSERC Undergraduate Summer Research Award. The work of D. Silva was supported by CAPES Foundation, Brazil. The material in this paper was presented in part at the 45th Annual Allerton Conference on Communication, Control, and Computing, Monticello, IL, September 2007.}
\thanks{D. Wang was with The Edward S. Rogers Sr. Department of Electrical and Computer Engineering, University of Toronto, Toronto, ON M5S 3G4, Canada. He is now with the Department of Electrical Engineering and Computer Science, Massachusetts Institute of Technology, Cambridge, MA 02139 USA (e-mail: dawang@mit.edu).}
\thanks{D. Silva was with The Edward S. Rogers Sr. Department of Electrical and Computer Engineering, University of Toronto, Toronto, ON M5S 3G4, Canada. He is now with the School of Electrical and Computer Engineering, State University of Campinas, Campinas, SP 13083-970, Brazil (e-mail: danilo@decom.fee.unicamp.br).}
\thanks{F. R. Kschischang is with The Edward S. Rogers Sr. Department of Electrical and Computer Engineering, University of Toronto, Toronto, ON M5S 3G4, Canada (e-mail: frank@comm.utoronto.ca).}
}

\begin{document}
\maketitle
\thispagestyle{empty}

\begin{abstract}
While network coding can be an efficient means of information
dissemination in networks, it is highly susceptible to ``pollution
attacks,'' as the injection of even a single erroneous packet has the
potential to corrupt each and every packet received by a given
destination.  Even when suitable error-control coding is applied, an
adversary can, in many interesting practical situations, overwhelm the
error-correcting capability of the code.  To limit the power of
potential adversaries, a broadcast transformation is introduced, in
which nodes are limited to just a single (broadcast) transmission per
generation.  Under this  broadcast transformation, the multicast
capacity of a network is changed (in general reduced) from the number of edge-disjoint paths between source and sink to the number of
internally-disjoint paths.  Exploiting this fact, a family of
networks is proposed whose capacity is largely unaffected by a broadcast
transformation. This results in a significant achievable transmission
rate for such networks, even in the presence of adversaries.
\end{abstract}

\begin{IEEEkeywords}
adversarial nodes,
broadcast transformation,
error correction,
JLC networks,
multicast capacity,
network coding.
\end{IEEEkeywords}

\section{Introduction}
\label{sec:introduction}

Network coding \cite{Ahlswede++2000} is a promising approach for
efficient information dissemination in packet networks.   Network coding
generalizes routing, allowing nodes in the network not only to switch
packets from input ports to output ports, but also to combine incoming
packets in some manner to form outgoing packets.  For example, in
\emph{linear} network coding, fixed-length packets are regarded as
vectors over a finite field $\Fq$, and nodes in the network form
outgoing packets as $\Fq$-linear combinations of incoming packets. For
the single-source multicast problem, it is known that linear network
coding suffices to achieve the network capacity
\cite{Li++2003,Koetter.Medard2003}.

Recently the problem of error correction in network coding has received
significant attention due to the fact that pollution attacks can be
catastrophic. Indeed, the injection of even a single erroneous packet
somewhere in the network has the potential to corrupt each and every
packet received by a given sink node.  This problem was first
investigated from an edge-centric perspective \cite{Cai.Yeung2002},
where a number of packet errors could arise in any of the links in the
network. Alternatively, under a node-centric perspective, it is assumed
that an adversarial node may join the network and transmit corrupt
packets on all its outgoing links, but the other links in the network
remain free of error.

One approach, investigated in \cite{Charles++2006,Zhao++2007}, for
dealing with the pollution problem is to apply cryptographic techniques
to ensure the validity of received packets, permitting corrupted packets
to be discarded by each node, and therefore preventing the contamination
of other packets.  This approach typically requires the use of large
field and packet sizes, which leads to computationally expensive
operations at the nodes and possibly to significant transmission delay.
These requirements may be acceptable in the large-file-downloading
scenario, but may be incompatible with delay-constrained applications
such as streaming-media distribution.

Another approach (and the one followed in this paper) is to look for
end-to-end coding techniques that require little or no intelligence at
the internal nodes. Jaggi \emph{et al.} \cite{Jaggi++2008} show that, if
$C$ is the network capacity (per transmission-generation) and $z$ is the
min-cut from the adversary to a destination, then a rate of $C - 2z$
packets per generation is achievable.  The same rate can also be
achieved using the subspace approach introduced by K{\"o}tter and
Kschischang \cite{Kotter.Kschischang2008,Silva++2008}. A higher rate
$C-z$ can be achieved using a scheme proposed in \cite{Jaggi++2008} (see
also \cite{Nutman.Langberg2008}) if the source and sink nodes are
allowed to share a secret (i.e., if they have common information not
available to the adversary).

In all of the end-to-end techniques mentioned above, we observe that the
min-cut from the adversary to a sink node has a significant impact on
the achievable rates. If $z$ is large---for instance, if $z=C$---then
the adversary can jam the network with no hope of recovery.  It is
important, therefore, to conceive of protocols that induce
per-generation network topologies that can perform well, even in the
presence of adversaries.

The central question of this paper is the following:
\begin{quote}
{\small\bf What simple changes to a protocol (and hence to
the induced network topology) might be effective in
reducing the influence of an adversary, while not
(greatly) affecting the rate of reliable communication?}
\end{quote}

We show that in some important special cases it is indeed possible to
constrict potential adversaries, without any sacrifice of network
capacity.

In this paper, we introduce the concept of a \emph{broadcast
transformation}, which essentially constrains potential adversaries to
sending the same packet on all its outgoing links. In the case of a
single malicious node, this effectively enforces $z=1$. In order for
such a transformation to be possible, we introduce the concept of a
\emph{trusted node} that performs the role of broadcasting traffic.  A
beneficial side-effect of a broadcast transformation is to lower the
encoding complexity, since each node only needs to compute a single
outgoing packet in each round of communication.

In practice, such a broadcasting feature could be implemented at
trusted network gateways. For example, in overlay network applications,
it could be implemented by ISPs at their gateways, through the use of deep packet inspection or similar technologies. Note that the broadcast constraint is effectively enforced if all packets in the same generation\footnote{Here we assume the use of generation-based network coding, as proposed in \cite{Chou++2003}.} from the same user have identical payload (although with different headers corresponding to different destination addresses). Thus, for each user/generation pair, the network gateway could simply store the payload of the first packet it receives and drop any subsequent packets that have different payloads (while also flagging such a user as ``suspicious''). It is worth mentioning that, for wireless networks, this constraint is automatically satisfied due to the broadcast nature of wireless communication~\cite{Dana++2006}, so the results of this paper are also naturally applicable in this case.

In general, a broadcast transformation can reduce capacity
(significantly, in some cases), unless the network has special
connectivity properties. We will show that the maximum number of
\emph{internally-disjoint paths} between source and sink, rather than
edge-disjoint paths, becomes the key parameter. This result implies that
robustness to node failures and robustness to adversarial attacks are
closely related concepts. We then examine a class of networks, which we
call \emph{$d$-diverse} networks, that have excellent robustness
properties. This class of networks is strongly inspired by the work of
Jain, Lov\'{a}sz and Chou in~\cite{Jain++2007:JLC} on robust and
scalable network topologies. We show that, under certain conditions, no
loss in capacity is incurred when performing broadcast conversion in
such $d$-diverse networks.

The remainder of this paper is organized as follows. In
Sec.~\ref{sec:preliminaries} we review some basic concepts of graph
theory and network coding. In Sec.~\ref{sec:adversarial} we introduce
an adversarial model for communication over untrusted networks.  In
Sec.~\ref{sec:broadcast} we introduce the broadcast transformation and
characterize the achievable rates of broadcast-constrained networks by
relating it to parameters of the original network. In
Sec.~\ref{sec:d-diverse-networks} we introduce $d$-diverse networks and
study their robustness properties. In Sec.~\ref{sec:conclusions} we present our conclusions.

\section{Preliminaries}
\label{sec:preliminaries}

\subsection{Graph Theory}
\label{ssec:graph-theory}

In this paper, a \emph{graph} always means a directed multigraph, i.e.,
all edges are directed and multiple edges between nodes\footnote{We will
use ``vertex'' and ``node'' interchangeably in this paper.}
are allowed. If
$\calG$ is a graph, then $\mathcal{V}(\mathcal{G})$ and
$\mathcal{E}(\mathcal{G})$ denote its vertex set and edge set,
respectively. Let $\mathbb{Z}_+ = \left\{ 1, 2, 3, \ldots \right\}$. We assume that
$\mathcal{E}(\mathcal{G}) \subseteq \mathcal{V}(\mathcal{G}) \times
\mathcal{V}(\mathcal{G}) \times \mathbb{Z}_+$,
where
the third component is used to distinguish among multiple edges between the same nodes.

For $\mathcal{A},\mathcal{B} \subseteq \mathcal{V}(\mathcal{G})$,
let $[\mathcal{A},\mathcal{B}]$ denote the set of all edges $(a,b,i)$ in $\mathcal{G}$ such that $a \in \mathcal{A}$ and $b \in \mathcal{B}$. We may also denote $[a,\mathcal{B}] \triangleq [\{a\},\mathcal{B}]$, $[\mathcal{A},b] \triangleq [\mathcal{A},\{b\}]$ and $[a,b] \triangleq [\{a\},\{b\}]$. For $[\calA,\calB]$ and any other concept that implicitly depends on $\calG$, we will use a subscript such as $[\calA,\calB]_\calG$ if the graph is not clear from the context.

If $\calS \subseteq \calV(\calG)$, then $\calG - \calS$ is the graph consisting of the vertex set $\calV(\calG) \setminus \calS$ and edge set $\calE(\calG) \setminus [\calV,\calS] \cup [\calS,\calV]$.

Let $|\calS|$ denote the cardinality of a set $\calS$. For nodes $u$ and $v$, if $|[u,v]|>0$, then $u$ is called a \emph{parent} of $v$, while $v$ is called a \emph{child} of $u$. We use $\Gamma^-(v)$ and $\Gamma^+(v)$ to denote, respectively, the set of all parents and the set of all children of a node $v$.

Let $\indegree(v)=|[\calV(\calG),v]|$ and $\outdegree(v)=|[v,\calV(\calG)]|$.

For $e \in [u,v]$, let $\tail(e) = u$ and $\head(e) = v$.
Also, for $\mathcal{E} \subseteq \mathcal{E}(\mathcal{G})$,
let $\tail(\mathcal{E}) \triangleq \cup_{e \in \mathcal{E}} \tail(e)$ and,
similarly, let $\head(\mathcal{E}) \triangleq \cup_{e \in \mathcal{E}} \head(e)$.

For $\mathcal{S} \subseteq \mathcal{V}(\mathcal{G})$, let
$\bar{\mathcal{S}} \triangleq \mathcal{V}(\mathcal{G}) \setminus
\mathcal{S}$.
For distinct nodes $s$ and $t$, if $s \in \mathcal{S}$ and $t \in \bar{\mathcal{S}}$, then $[\mathcal{S},\bar{\mathcal{S}}]$ is called an \emph{$s,t$-edge cut}. Let
\begin{equation}\nonumber
  \K(s,t) \triangleq \min_{\substack{\mathcal{S} \subseteq \mathcal{V}(\mathcal{G}):\\ s \in \mathcal{S} \not\ni t}}\, |[\mathcal{S},\bar{\mathcal{S}}]|
\end{equation}
denote the minimum size of an $s,t$-edge cut. Note that $\K(s,t)$ is often denoted by $\kappa'(s,t)$. For convenience, define also
\begin{equation}\nonumber
  \K(\mathcal{A},t) \triangleq \min_{\substack{\mathcal{S} \subseteq \mathcal{V}(\mathcal{G}):\\ \mathcal{A} \subseteq \mathcal{S} \not\ni t }}\, |[\mathcal{S},\bar{\mathcal{S}}]|.
\end{equation}

A \emph{path} is a sequence of vertices such that from each vertex
there is an edge to the next vertex in the sequence.
The first and last vertices in a finite path are called \emph{end
vertices}, and the other vertices are
called \emph{internal vertices}.

For distinct nodes $s$ and $t$, a set $\calS \subseteq \calV(\calG)\setminus \{s,t\}$ is called an \emph{$s,t$-vertex cut} if $\calG - \calS$ has no path connecting $s$ and $t$.
Note that for an $s,t$-vertex cut to exist, $t$ cannot be a child of $s$. In that condition, let $\kappa(s,t)$ denote the minimum size of an $s,t$-vertex cut.

Two paths are called \emph{edge-disjoint} if they have no edges in
common, and are called \emph{internally-disjoint} if they have no
internal nodes in common. Let $\lambda'(s,t)$ denote the maximum number of pairwise edge-disjoint paths from a node $s$ to a node $t$ and let $\lambda(s,t)$ denote the maximum number of pairwise internally-disjoint paths from $s$ to $t$.

We will frequently refer to the edge and vertex versions of Menger's Theorem on directed graphs~\cite{West} (the former is also known as the Max-flow Min-cut Theorem).
\medskip
\begin{theorem}[Menger's Theorem, edge version]
    \label{thm:menger-edge}
    For any vertices $s$ and $t$, $\lambda'(s,t) = \K(s,t)$.
\end{theorem}
\medskip
\begin{theorem}[Menger's Theorem, vertex version]
    \label{thm:menger-vertex}
    For any vertices $s$ and $t$, if $|[s,t]|=0$, then $\lambda(s,t) = \kappa(s,t)$.
\end{theorem}

\subsection{Network Coding}
\label{ssec:network-coding}

A \emph{(single-source) multicast network} $\mathcal{N} = (\mathcal{G},s,\mathcal{T})$
consists of a (directed multi)graph $\mathcal{G}$ with a distinguished
\emph{source node}~$s$ and a set of \emph{sink nodes} $\mathcal{T}
\not\ni s$. In a multicast problem, each sink node requests the same
message that is observed at the source node.

Each link in the network is assumed to transport, free of errors, a
packet of a certain fixed size. A packet in a link entering a node is
said to be an incoming packet to that node, and similarly a packet in a
link leaving a node is said to be an outgoing packet from that node.

When network coding is used, the source node produces each of its
outgoing packets as an arbitrary function of the message it observes.
Also, each non-source node produces each of its outgoing packets as an
arbitrary function of its incoming packets. The set of functions
applied by all nodes in the network specifies a \emph{network code}. If
each sink node, by observing its incoming packets, is able to correctly
identify the source message, then we say that the decoding is
successful.

Let $q$ be the size of the set from which packets are selected and let $\Omega$ be the set from which the source message is
selected. The \emph{rate} of communication is defined as
\begin{equation}\nonumber
  R(\Omega,q) \triangleq \log_{q} |\Omega|
\end{equation}
which is the amount of information, measured in packets, that can be conveyed by the source message.

A rate $R$ is said to be \emph{achievable} for a network $\mathcal{N}$
if, for any $\epsilon > 0$, there exist $q$ and $\Omega$ with $R(\Omega,q) \geq R$,
along with a corresponding network code, such that the probability of unsuccessful decoding is smaller than $\epsilon$.

For a multicast network $\mathcal{N} = (\mathcal{G},s,\mathcal{T})$, define
\begin{equation}\nonumber
  C(\mathcal{N}) \triangleq \min_{t \in \mathcal{T}}\, \K_\mathcal{G}(s,t).
\end{equation}
A key result in \cite{Ahlswede++2000} is that a rate $R$ is achievable for $\mathcal{N}$ if and only if
\begin{equation}\nonumber
  R \leq C(\mathcal{N}).
\end{equation}
For this reason, $C(\mathcal{N})$ is referred to as the \emph{capacity} of a multicast network $\mathcal{N}$.

\section{Untrusted Multicast Networks}
\label{sec:adversarial}

In this section we describe a node-centric adversarial model for networks that can be subject to pollution attacks. This model will be used in the remainder of the paper for the computation of achievable rates.

We start with the definition of an untrusted multicast network. Consider a multicast network. A node is said to be \emph{trusted} if it is guaranteed to behave according to a specified network coding protocol; otherwise, it is said to be \emph{untrusted}. In particular, a trusted node cannot be controlled by an adversary, while an untrusted node may (or may not) be so. An \emph{untrusted multicast network} $\mathcal{N} = (\mathcal{G},s,\mathcal{T},\mathcal{U})$ is a multicast network $(\mathcal{G},s,\mathcal{T})$ with a specified set of untrusted nodes $\mathcal{U} \subseteq \mathcal{V}(\mathcal{G}) \setminus \{s\}$ such that all nodes in $\mathcal{V}(\mathcal{G})\setminus \mathcal{U}$ are trusted.

An adversarial model for communication over an
untrusted multicast network may be specified as follows. The adversary chooses a set of adversarial
nodes $\mathcal{A} \subseteq \mathcal{U}$ with $|\mathcal{A}| \leq w$
prior to the beginning of the session. The set $\mathcal{A}$ is unknown
to source and sink nodes, but remains fixed during the whole session.
The adversary controls the nodes in $\mathcal{A}$, which are allowed to
transmit any arbitrary packets on their outgoing links and also to
cooperate with each other. Since an adversarial node cannot be counted
as a sink node, we say that decoding is successful if each node in
$\mathcal{T}\setminus \mathcal{A}$ can correctly recover the source
message.

Several end-to-end error control schemes have been proposed to ensure
reliable communication over an untrusted network
\cite{Jaggi++2008,Kotter.Kschischang2008,Silva++2008,Nutman.Langberg2008,Silva.Kschischang2009:KeyBased-CWIT}.
The rates achievable by these schemes depend on further assumptions on
the system model. In this paper, we focus on the two most basic of these
models. The \emph{omniscient adversary} (OA) model refers to the case
where no constraints are imposed on the knowledge or computational power
of the adversary. If an additional assumption is made that common
randomness is available between the source and sink nodes, then
resulting scenario is called the \emph{shared secret} (SS) model.

Achievable rates under these models are often stated from an
edge-centric perspective, i.e., assuming that the adversary controls a
certain number of edges. Below we restate these results from a
node-centric perspective.

\medskip
\begin{theorem}[\cite{Jaggi++2008,Silva++2008}]\label{thm:rate-OA}
Let $\mathcal{N} = (\mathcal{G},s,\mathcal{T},\mathcal{U})$ be an untrusted multicast network with at most $w$ adversarial nodes. Under the shared secret model, it is possible to achieve the rate
\begin{equation}\label{eq:rate-OA}
    R^{\sf OA}(\mathcal{N},w) \triangleq
        \min_{\substack{\mathcal{A} \subseteq \mathcal{U}:\\|\mathcal{A}| \leq w}}\,
            \min_{t \in \mathcal{T} \setminus \mathcal{A}}\,
                R^{\sf OA}(s,t,\mathcal{A})
\end{equation}
where
\begin{equation}\label{eq:rate-OA-individual}\nonumber
  R^{\sf OA}(s,t,\mathcal{A}) \triangleq \left[ \K(s,t) - 2\K(\mathcal{A},t) \right]^+.
\end{equation}
\end{theorem}
\medskip

\medskip
\begin{theorem}[\cite{Jaggi++2008,Nutman.Langberg2008,Silva.Kschischang2009:KeyBased-CWIT}]\label{thm:rate-SS}
Let $\mathcal{N} = (\mathcal{G},s,\mathcal{T},\mathcal{U})$ be an untrusted multicast network with at most $w$ adversarial nodes. Under the omniscient adversary model, it is possible to achieve the rate
\begin{equation}\label{eq:rate-SS}
    R^{\sf SS}(\mathcal{N},w) \triangleq
        \min_{\substack{\mathcal{A} \subseteq \mathcal{U}:\\|\mathcal{A}| \leq w}}\,
            \min_{t \in \mathcal{T} \setminus \mathcal{A}}\,
                R^{\sf SS}(s,t,\mathcal{A})
\end{equation}
where
\begin{equation}\label{eq:rate-SS-individual}\nonumber
  R^{\sf SS}(s,t,\mathcal{A}) \triangleq \left[ \K(s,t) - \K(\mathcal{A},t) \right]^+.
\end{equation}
\end{theorem}
\medskip

We will use (\ref{eq:rate-OA}) and (\ref{eq:rate-SS}) as benchmarks to
evaluate the effective throughput of a multicast network in the presence
of adversaries.

Note that when there is no adversary, both expressions
reduce to the capacity of the underlying network, i.e.,
\begin{equation}\nonumber
  R^{\sf OA}(\mathcal{N},0) = R^{\sf SS}(\mathcal{N},0) = C(\mathcal{N}).
\end{equation}

From Theorems~\ref{thm:rate-OA} and~\ref{thm:rate-SS} we observe that, for an
adversarial set $\mathcal{A}$ and a sink node $t$, the quantity
$\K(\mathcal{A},t)$ can have a severe impact on the achievable rate of
the untrusted network. If $\K(\mathcal{A},t)$ is large compared to
$\K(s,t)$, then the adversary can overwhelm the system with corrupt
packets, preventing successful decoding.

\section{Broadcast Transformation}
\label{sec:broadcast}

In this section, we propose an approach to restrict the min-cut between
adversarial nodes and sink nodes, which can lead to potentially higher
achievable rates over untrusted networks.  The idea is to force each
adversarial node to transmit only copies of the same packet, effectively
constraining its outdegree to be at most 1. As we do not know beforehand
which nodes are adversarial, the constraint must be enforced on every
\emph{untrusted} node.
This operation can be represented graphically by introducing a new node
$u^+$, as described in Fig.~\ref{fig:broadcast}.
\begin{figure}
  \centering
  \includegraphics[scale=0.6]{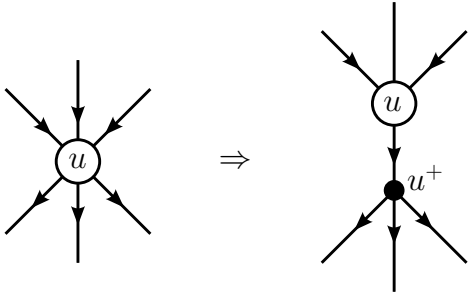}
%  \scalebox{0.6}{\input{figs/bc_transform.pstex_t}}
  \caption{Broadcast transformation.}
  \label{fig:broadcast}
\end{figure}
Here, $u^+$ is a \emph{trusted node} that only replicates the packet
received. The overall operation, which we refer to as a \emph{broadcast transformation}, is formally defined below.
\smallskip
\begin{definition}\label{def:broadcast-transformation}
Let $\mathcal{N} = (\mathcal{G},s,\mathcal{T},\mathcal{U})$ be an untrusted multicast network with $\mathcal{G}= (\mathcal{V},\mathcal{E})$. The \emph{broadcast transformation} of $\calN$, denoted by $\beta(\mathcal{N})$, is an untrusted multicast network $(\hat{\mathcal{G}},s,\mathcal{T},\mathcal{U})$ with
$\hat{\mathcal{G}} = (\hat{\mathcal{V}},\hat{\mathcal{E}})$ given by
\begin{align}
    \hat{\mathcal{V}} &= \mathcal{V} \cup \{u^+ \colon u \in \mathcal{U}\} \nonumber \\
    \hat{\mathcal{E}} &= (\mathcal{E} \setminus
    [\mathcal{U},\mathcal{V}]) \cup \{(u,u^+,1)\colon u \in \mathcal{U}\} \cup [\mathcal{U},\mathcal{V}]^+ \nonumber
\end{align}
where $[\mathcal{U},\mathcal{V}]^+ = \{(u^+,v,i) \colon (u,v,i) \in [\mathcal{U},\mathcal{V}]\}$.
\end{definition}
\medskip

After a broadcast transformation, adversarial nodes can only do limited
harm, as shown in the following simple result.

\medskip
\begin{proposition}\label{prop:broadcast-rates}
  Let $\beta(\calN)$ be the broadcast transformation of an untrusted multicast network $\mathcal{N} =(\mathcal{G},s,\mathcal{T},\mathcal{U})$. For $0 \leq w \leq C(\beta(\calN))$, we have
  \begin{align}
  R^{\sf OA}(\beta(\mathcal{N}),w) &\geq \left[C(\beta(\mathcal{N})) - 2 w\right]^+  \nonumber \\
  R^{\sf SS}(\beta(\mathcal{N}),w) &\geq \left[C(\beta(\mathcal{N})) - w\right]^+  \nonumber
  \end{align}
  with equality if $\mathcal{U} = \mathcal{V}(\mathcal{G})\setminus \{s\}$.
\label{prop:rate-r-transform}
\end{proposition}
\begin{proof}
Let $(\hat{\mathcal{G}},s,\mathcal{T},\mathcal{U}) = \beta(\calN)$.
The pair of inequalities follows immediately from Definition~\ref{def:broadcast-transformation} and Theorems~\ref{thm:rate-OA} and~\ref{thm:rate-SS} by noticing that $\K_{\hat{\mathcal{G}}}(\mathcal{A},t) \leq |\mathcal{A}|$ for any $\mathcal{A} \subseteq \mathcal{U}$ and any $t \in \mathcal{T} \setminus \mathcal{A}$.

For the case $\mathcal{U} = \mathcal{V}(\mathcal{G})\setminus \{s\}$, let $t \in \mathcal{T}$ be any node satisfying $\K_{\hat{\mathcal{G}}}(s,t) = C(\beta(\mathcal{N}))$. Note that $t$ must have at least $C(\beta(\mathcal{N}))$ distinct parents in $\mathcal{G}$, all of which are untrusted. Take any $w$ of such parents to form a set $\mathcal{A}$. Then $\K_{\hat{\mathcal{G}}}(\mathcal{A},t) = w$, which shows that both inequalities can be met with equality.
\end{proof}
\medskip

In general, applying a broadcast transformation may reduce
$C(\beta(\mathcal{N}))$, the multicast capacity of the resulting
network. Still, the reduction in the jamming capability of the adversary
may compensate for this loss and yield a higher achievable rate. This
trade-off, which is captured by Proposition~\ref{prop:broadcast-rates},
will be shown to be indeed favorable in certain meaningful situations.
More specifically, we are interested in studying networks for which
$C(\beta(\mathcal{N}))$ is equal or approximately equal to
$C(\mathcal{N})$. If this is the case, we will say that $\mathcal{N}$ is
a \emph{robust} network.

In the remainder of the paper, we restrict attention to the case
$\mathcal{U} = \mathcal{V}(\mathcal{G})\setminus \{s\}$, where all
non-source nodes are untrusted. This case not only has analytical
advantages, but also seems to be the case of most practical relevance.

For a multicast network $\mathcal{N} =(\mathcal{G},s,\mathcal{T})$, define
\begin{equation*}
  \Lambda(\mathcal{N}) \triangleq \min_{t \in
  \mathcal{T}}\,\lambda_\mathcal{G}(s,t).
\end{equation*}
The following theorem shows that the multicast capacity of a broadcast-transformed network has a nice graph-theoretical characterization in terms of the original network.

\begin{theorem}\label{thm:rate-after-broadcast}
Let $\mathcal{N} =(\mathcal{G},s,\mathcal{T},\mathcal{U})$ be an
untrusted multicast network with
$\mathcal{U} = \mathcal{V}(\calG) \setminus \{s\}$.
Then
\begin{equation}\nonumber
  C(\beta(\mathcal{N})) = \Lambda(\mathcal{N}).
\end{equation}
\end{theorem}
\begin{proof}
The proof is closely related to the standard argument used to derive Theorem~\ref{thm:menger-vertex}
from the Max-flow Min-cut Theorem.

Let $\beta(\calN) = (\hat{\calG},s,\calT,\calU)$.
Since $\mathcal{U} = \mathcal{V}(\calG) \setminus \{s\}$, the broadcast
transformation replaces each non-source node by a node followed by an
edge followed by a node, as illustrated in Fig~\ref{fig:broadcast}.
Thus, if two paths in $\mathcal{G}$ are internally-disjoint, then they
will also be internally- (and therefore edge-) disjoint in
$\hat{\mathcal{G}}$. Conversely, if two paths in $\mathcal{G}$ are not
internally-disjoint, i.e., they share a node $v$, then they will also
share the two nodes $v$ and $v^+$ and the edge $(v,v^+,1)$ in
$\hat{\mathcal{G}}$, and therefore will not be edge-disjoint in
$\hat{\mathcal{G}}$. Thus, for any $t \in \calT$, the maximum number of
internally-disjoint paths from $s$ to $t$ in $\mathcal{G}$ must be equal to the maximum
number of edge-disjoint paths from $s$ to $t$ in $\hat{\mathcal{G}}$,
i.e.,
$\lambda_\calG(s,t) = \lambda'_{\hat{\calG}}(s,t) = \K_{\hat{\mathcal{G}}}(s,t)$.
The result now follows from the
definitions of $\Lambda(\calN)$ and $C(\beta(\calN))$.
\end{proof}
\medskip

We now give some examples of robust and non-robust networks.

\medskip
\begin{example}\label{eg:nonrobust}
Consider the network $\mathcal{N}$ in Fig.~\ref{fig:nonrobust_eg},
\begin{figure}
  \centering
  \includegraphics[scale=0.6]{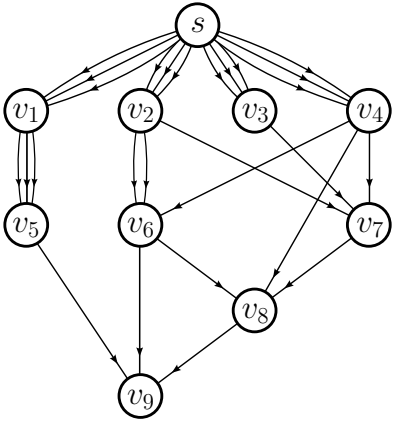}
%  \scalebox{0.6}{\input{figs/nonrobust_eg.pstex_t}}
	\caption{A non-robust %untrusted multicast
  network with $C(\mathcal{N}) = 3$
	and $C(\beta(\mathcal{N})) = 1$.}
  \label{fig:nonrobust_eg}
\end{figure}
where $s$ is the source node and all other nodes $v_1,\ldots,v_9$ are untrusted sink nodes. Note that, for any $v_i$, we have $\K(s,v_i) = 3$, and therefore  $C(\mathcal{N}) = 3$. Meanwhile, $\lambda(s,v_5) = 1$, so $C(\beta(\mathcal{N})) = \Lambda(\mathcal{N}) = 1$. Thus, $\mathcal{N}$ is not a robust network.
\end{example}

\medskip
\begin{example}\label{eg:robust}
  To make the network in Fig.~\ref{fig:nonrobust_eg} robust, we can increase the diversity of internally-disjoint paths to $v_5$ and $v_6$ by letting $v_5$ and $v_6$ have multiple parents. This may result in a network $\mathcal{N}$ as shown in Fig.~\ref{fig:robust_eg2}.
\begin{figure}
  \centering
  \includegraphics[scale=0.6]{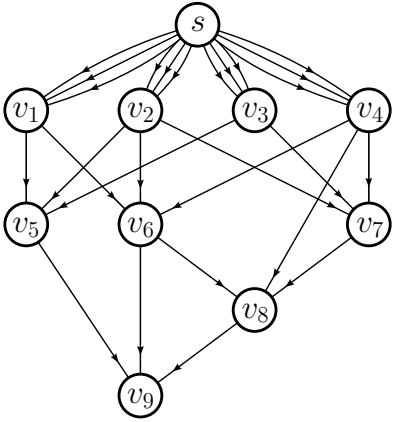}
%  \scalebox{0.6}{\input{figs/robust_eg.pstex_t}}
	\caption{A robust %untrusted multicast
  network with $C(\beta(\mathcal{N})) = C(\mathcal{N}) = 3$.}
  \label{fig:robust_eg2}
\end{figure}
  Now, for all $i$, we have $\K(s,v_i) = 3$ and $\lambda(s,v_i) = 3$.  Thus $C(\mathcal{N}) = 3$ and $C(\beta(\mathcal{N})) = \Lambda(\mathcal{N}) = 3$. Therefore, $\mathcal{N}$ is a robust network.
\end{example}

\section{$d$-diverse Networks}
\label{sec:d-diverse-networks}
In this section, we study a special class of networks, which we call \emph{$d$-diverse
networks}, that have simultaneously good capacity and robustness properties.
This class of networks is motivated by the notion of parent diversity illustrated in Example~\ref{eg:robust}.

\medskip
\begin{definition}\label{def:d-diverse}
Let $\mathcal{N} = (\mathcal{G},s,\mathcal{T})$ be an acyclic multicast network. The \emph{(parent) diversity} of a non-source node $v \in \calV(\calG)\setminus \{s\}$ is defined as
\begin{equation}\nonumber
  d(v) \triangleq |\Gamma^-(v)\setminus\{s\}| + |[s,v]|.
\end{equation}
The \emph{(parent) diversity} of $\calN$ is defined as
\begin{equation}\nonumber
  d(\calN) \triangleq \min_{v \in \calV(\calG)\setminus \{s\}}\, d(v).
\end{equation}
If $d(\calN) = d$, then $\calN$ is called a \emph{$d$-diverse network}.
\end{definition}
\medskip

For any node that is nonadjacent to the source node, the parent diversity is exactly the cardinality of its parent set. For a node that is adjacent to the source node, this interpretation remains true if we replace each edge coming from the source node by an edge followed by a node followed by an edge. This slight twist in the definition is required due to the special role that a source node has in a network problem.

The following is the main result of this section.

\medskip
\begin{theorem}\label{thm:d-diverse-lambda}
Let $\calN = (\calG,s,\calT)$ be an acyclic network. Then
\begin{equation}\nonumber
  \Lambda(\mathcal{N}) \geq d(\calN).
\end{equation}
In particular, if $\indegree(t) = d(\calN)$ for some $t \in \calT$, then
\begin{equation}\nonumber
%  d(\calN) = \Lambda(\calN) = C(\mathcal{N}).
  C(\calN) = \Lambda(\calN) = d(\mathcal{N}).
\end{equation}
\end{theorem}
\medskip

Theorem~\ref{thm:d-diverse-lambda} shows that, for large enough $d$, a
$d$-diverse network not only has good multicast capacity but is also
robust. In particular, when designing a network, one might focus solely
on achieving high parent diversity, obtaining good capacity and
robustness as natural consequences. It is important to note that, while
$C(\calN)$ and $\Lambda(\calN)$ are global parameters of the network,
the diversity $d(\calN)$ (or rather $d(v)$ for each node $v$) is a
parameter that depends only on local information available at a node.
Therefore, it should be relatively easy to construct a $d$-diverse
network by enforcing $d(v) \geq d$ at each node. This is indeed the case for the class of JLC networks, as discussed later in Example~\ref{eg:JLC}.

In order to prove Theorem~\ref{thm:d-diverse-lambda}, we start with a
lemma that characterizes minimal vertex cuts in a graph.

\medskip
\begin{lemma}\label{lem:minimal-vertex-cut}
Consider a graph $\mathcal{G} = (\mathcal{V},\mathcal{E})$ with
nonadjacent nodes $s$ and $t$. Then every minimal $s,t$-vertex cut is given by $\tail([\mathcal{S},\bar{\mathcal{S}}])$ for some $s,t$-edge cut $[\mathcal{S},\bar{\mathcal{S}}]$.
In particular,
\begin{equation}\label{eq:lambda-min-tail}
    \lambda_\mathcal{G}(s,t) = \min_{[\calS,\bar{\calS}]}\, |\tail([\mathcal{S},\bar{\mathcal{S}}])|
\end{equation}
where the minimization is taken over all $s,t$-edge cuts
$[\mathcal{S},\bar{\mathcal{S}}]$ such that
$s \not\in \tail([\mathcal{S},\bar{\mathcal{S}}])$.
\end{lemma}
\begin{proof}
First, note that if $[\mathcal{S},\bar{\mathcal{S}}]$ is an $s,t$-edge cut such that $s \not\in \tail([\mathcal{S},\bar{\mathcal{S}}])$, then
$\tail([\mathcal{S},\bar{\mathcal{S}}])$ is indeed an $s,t$-vertex cut.
This follows from the fact that removing
$\tail([\mathcal{S},\bar{\mathcal{S}}])$ from $\calG$ also removes all
the edges in $[\mathcal{S},\bar{\mathcal{S}}]$.

We now show that if $\mathcal{A}$ is a minimal $s,t$-vertex cut, then
there exists some $s,t$-edge cut $[\mathcal{S},\bar{\mathcal{S}}]$
with $s \not\in \tail([\mathcal{S},\bar{\mathcal{S}}])$ such that
$\mathcal{A} = \tail([\mathcal{S},\bar{\mathcal{S}}])$. For this, consider the graph $\calG - \calA$. Since $\mathcal{A}$ is an $s,t$-vertex cut, the graph $\mathcal{G} - \mathcal{A}$ has two components. Let $\mathcal{A}_{s}$ and
$\mathcal{A}_{t}$ be the components that contain $s$ and $t$,
respectively. Let $\mathcal{S} = \mathcal{A}_s \cup \mathcal{A}$; then
$\bar{\mathcal{S}} = \mathcal{A}_t$. Note that
$[\mathcal{S},\bar{\mathcal{S}}]$ is an $s,t$-edge cut. Moreover,
$\tail([\mathcal{S},\bar{\mathcal{S}}]) \subseteq \mathcal{A}$,
otherwise $\mathcal{A}$ would not separate $s$ and $t$. Since
$\tail([\mathcal{S},\bar{\mathcal{S}}])$ is also an $s,t$-vertex cut and
$\mathcal{A}$ is minimal, we conclude that
$\tail([\mathcal{S},\bar{\mathcal{S}}]) = \mathcal{A}$. In addition, we
must have $[\mathcal{S},\bar{\mathcal{S}}] \cap [s,\mathcal{V}] =
\emptyset$, otherwise $s \in \mathcal{A}$, which is impossible by the
definition of an $s,t$-vertex cut.

Now the result follows immediately from
Theorem~\ref{thm:menger-vertex}.
\end{proof}
\medskip

We can now give a proof of Theorem~\ref{thm:d-diverse-lambda}.

\medskip
\begin{proof}[Proof of Theorem~\ref{thm:d-diverse-lambda}]
Let $t \in \mathcal{T}$. First, suppose $t$ is not adjacent to $s$. Let $[\mathcal{S},\bar{\mathcal{S}}]$ be some $s,t$-edge cut achieving the minimization in (\ref{eq:lambda-min-tail}). Since the graph $\calG$ is directed acyclic, it has at least one topological ordering. Let $u$ be the first node in $\bar{\calS}$ according to some topological ordering, i.e., $u \in \bar{\calS}$ is a node whose parents are all in $\calS$. We have
\begin{align}
\lambda_\mathcal{G}(s,t)
&= |\tail([\mathcal{S},\bar{\mathcal{S}}])| \nonumber \\
&\geq |\Gamma^-(u)| \nonumber \\
&\geq d(\calN) \label{eq:proof-lambda-d-diverse-3}
\end{align}
where
(\ref{eq:proof-lambda-d-diverse-3}) follows from the fact that $|[s,u]| = 0$, since $s \not\in \tail([\mathcal{S},\bar{\mathcal{S}}])$.

Now, suppose $t$ is adjacent to $s$. Let $m = |[s,t]|$. Consider a new network $\mathcal{N}' = (\mathcal{G}',s,\mathcal{T})$, where $\mathcal{G}' = \mathcal{G} - [s,t]$. Note that $d(\calN') \geq d(\calN) - m$. Using the argument above on $\mathcal{N}'$, we obtain that
\begin{equation}\nonumber
  \lambda_{\mathcal{G}'}(s,t) \geq d(\calN') \geq d(\calN) - m.
\end{equation}
Returning to the original network, we have
\begin{equation}\nonumber
  \lambda_{\mathcal{G}}(s,t) = \lambda_{\mathcal{G}'}(s,t) + m \geq d(\calN).
\end{equation}

From the above arguments, it follows that $\Lambda(\mathcal{N}) \geq d(\calN)$. The special case follows immediately since $\Lambda(\calN) \leq C(\calN) \leq \indegree(t)$, for all $t \in \calT$.
\end{proof}
\medskip

As an application of Theorem~\ref{thm:d-diverse-lambda}, consider the
case of a network in which all non-source nodes are sink nodes with
diversity exactly $d$, and such that there are no parallel edges between
nodes, except possibly emanating from the source node. Then the
multicast capacities both before and after broadcast transformation are
exactly equal to $d$. Note that, as the indegree of any non-source node
is exactly $d$, any removed edge would result in a smaller capacity.
Thus, we may conclude that, given a fixed number of edges, the network
capacity is maximized by having nodes select incoming edges from
distinct parents rather than from the same parent. This result holds
even if all non-source nodes are untrusted, provided a broadcast
transformation is performed.

\medskip
\begin{example}[JLC networks]\label{eg:JLC}
We now describe a class of networks that has not only good theoretical
properties but also potential for practical applications. The protocol
for constructing and operating these networks has been proposed by Jain,
Lov{\'a}sz and Chou \cite{Jain++2007:JLC} as a scalable and robust
solution to peer-to-peer data dissemination with network coding. We
refer to any network constructed according to their protocol as a
\emph{JLC network}.

An example of a JLC network is depicted in Fig.~\ref{fig:JLC_eg}.
\begin{figure}
    \centering
    \includegraphics[scale=0.6]{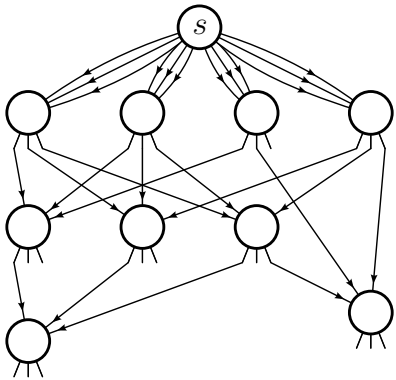}
%    \scalebox{0.60}{\input{figs/JLC_eg.pstex_t}}
    \caption{A $d$-diverse JLC network with $k=12$ and $d=3$.}
    \label{fig:JLC_eg}
\end{figure}
The network is acyclic, and all non-source nodes are sinks. Initially,
the network contains only the source node (or server), which has $k$
(potential) outgoing links. Here, each link represents a stream of unit
bandwidth. At any time, the server maintains a list of $k$ available
links for download. When a new node joins the network, it requests from
the server $d$ download links. The server randomly picks $d$ links from
the pool of available links, and updates its list with $d$ potential
links originating from the new node. Therefore, the network always has
$k$ links (i.e., streams of unit bandwidth) available for download.

It is easy to ensure that a JLC network is $d$-diverse by performing a simple protocol modification.
When a new node joins the network, rather than choosing the $d$ upstream
links completely at random from the $k$ available links (thereby
allowing the possibility of fewer than $d$ distinct parents), the server
simply needs to provide the new node with $d$ links from $d$ distinct
parents. Note that, in practice, $k \gg d^2$, so
the $k$ available links come from at
least $l=\lceil k/d \rceil \gg d$ parents.
Hence, the modification can be done easily.
\end{example}

\section{Conclusions}
\label{sec:conclusions}

We have introduced the broadcast transformation of a network, which
restricts the influence of potential adversaries by limiting them to a
single transmission opportunity per generation. For networks with a
sufficient diversity of internally-disjoint paths from source to
sink(s), the multicast capacity may not be greatly affected by this
transformation. In particular, for a class of networks called $d$-diverse networks, the full capacity is maintained when $d$ is sufficiently large.
Combined with error control for network coding, the proposed approach
may be an effective means of dealing with adversaries, particularly in
application scenarios such as real-time media streaming, where
alternative (e.g., cryptographic) methods may be cost-prohibitive.

\section*{Acknowledgements}

The authors would like to thank the anonymous reviewers for their helpful comments, which significantly improved the presentation of the paper.

\begin{IEEEbiographynophoto}{Da Wang}
received the B.A.Sc (Hons.) degree in electrical engineering from the University of Toronto, Toronto, ON, Canada, in 2008.

He is currently working toward the M.S. degree in the Department of Electrical Engineering and Computer Science (EECS) at the Massachusetts Institute of Technology (MIT), Cambridge. His research interests lie in the areas of communication and information theory.
\end{IEEEbiographynophoto}

\begin{IEEEbiographynophoto}{Danilo Silva}
(S'06--M'09) received the B.Sc. degree from the Federal University of Pernambuco, Recife, Brazil, in 2002, the M.Sc. degree from the Pontifical Catholic University of Rio de Janeiro (PUC-Rio), Rio de Janeiro, Brazil, in 2005, and the Ph.D. degree from the University of Toronto, Toronto, Canada, in 2009, all in electrical engineering.

From September to October 2009, he was a Postdoctoral Fellow with the Ecole Polytechnique F\'{e}d\'{e}rale de Lausanne (EPFL), and from October to December 2009, he was a Postdoctoral Fellow with the University of Toronto. He is currently a Postdoctoral Fellow with the State University of Campinas (Unicamp). His research interests include channel coding, information theory, and network coding.
\end{IEEEbiographynophoto}

\begin{IEEEbiographynophoto}{Frank R. Kschischang}
(S'83--M'91--SM'00--F'06) received the B.A.Sc. degree (with honors) from the University of British Columbia, Vancouver, BC, Canada, in 1985 and the M.A.Sc.\  and Ph.D.\ degrees from the University of Toronto, Toronto, ON, Canada, in 1988 and 1991, respectively, all in electrical engineering.  He is a Professor of Electrical and Computer Engineering and Canada Research Chair in Communication Algorithms at the University of Toronto, where he has been a faculty member since 1991. During 1997-98, he was a visiting scientist at MIT, Cambridge, MA and in 2005 he was a visiting professor at the ETH, Zurich.  His research interests are focused on the area of channel coding techniques.

He is the recipient of the Ontario Premier's Research Excellence Award, a Canada Council of the Arts Killam Research Fellowship, and (with R. Koetter) the IEEE Communications Society and Information Theory Society Joint Paper Award.

During 1997-2000, he served as an Associate Editor for Coding Theory for the IEEE TRANSACTIONS ON INFORMATION THEORY. He also served as technical program co-chair for the 2004 IEEE International Symposium on Information Theory (ISIT), Chicago, and as general co-chair for ISIT 2008, Toronto.  He serves as the 2010 President of the IEEE Information Theory Society.
\end{IEEEbiographynophoto}

\end{document}